\documentclass[aps,pra,twocolumn,superscriptaddress]{revtex4-2}
\usepackage{comment}
\usepackage{amsthm}
\usepackage{amssymb}
\usepackage{amsmath}
\usepackage{physics}
\usepackage{float}
\usepackage{graphicx}
\usepackage{dcolumn}
\usepackage{bm}
\usepackage[hidelinks]{hyperref}

\newcommand{\be}{\begin{equation}}
\newcommand{\ee}{\end{equation}}
\newcommand{\ba}{\begin{aligned}}
\newcommand{\ea}{\end{aligned}}

\newcommand{\R}{\mathbb{R}}
\newcommand{\bc}{\begin{center}}
\newcommand{\ec}{\end{center}}
\newcommand{\beq}{\begin{equation}}
\newcommand{\eeq}{\end{equation}}
\newcommand{\beqq}{\begin{equation*}}
\newcommand{\eeqq}{\end{equation*}}
\newcommand{\beqa}{\begin{align}}
\newcommand{\eeqa}{\end{align}}
\newcommand{\barr}{\begin{array}}
\newcommand{\earr}{\end{array}}
\newcommand{\bi}{\begin{itemize}}
\newcommand{\ei}{\end{itemize}}

\newcommand{\C}{\mathbb{C}}

\newtheorem{lem}{Lemma}
\newtheorem{theo}{Theorem}

\DeclareMathOperator{\N}{\mathbb{N}}
\DeclareMathOperator{\I}{\mathbb{I}}

\begin{document}


\title{Interplay of resources for universal continuous-variable quantum computing}

\author{Varun Upreti}
 \email{varun.upreti@inria.fr}
\author{Ulysse Chabaud}%
\email{ulysse.chabaud@inria.fr}
\affiliation{%
 DIENS, \'Ecole Normale Sup\'erieure, PSL University, CNRS, INRIA, 45 rue d’Ulm, Paris, 75005, France
}%

\date{\today}

\begin{abstract}
    Quantum resource theories identify the features of quantum computers that provide their computational advantage over classical systems. We investigate the resources driving the complexity of classical simulation in the standard model of continuous-variable quantum computing, and their interplay enabling computational universality. Specifically, we uncover a new property in continuous-variable circuits, analogous to coherence in discrete-variable systems, termed \textit{symplectic coherence}. Using quadrature propagation across multiple computational paths, we develop an efficient classical simulation algorithm for continuous-variable computations with low symplectic coherence. This establishes symplectic coherence as a necessary resource for universality in continuous-variable quantum computing, alongside non-Gaussianity and entanglement. Via the Gottesman--Kitaev--Preskill encoding, we show that the interplay of these three continuous-variable quantum resources mirrors the discrete-variable relationship between coherence, magic, and entanglement.
\end{abstract}

\maketitle



\section{Introduction}\label{sec:introduction}

Universal quantum computers may offer significant advantages over classical counterparts \cite{shor1994algorithms,Aaronson2013,Ronnow2014,Harrow2017,Boixo2018}. Understanding the origin of these advantages is crucial for advancing quantum technologies. This requires identifying scenarios where classical simulation of quantum computations—replicating quantum behavior with classical algorithms—becomes inefficient \cite{Shi2006,Huang2015,Yung2019,Zhou2020,Xu2023}. Quantum resource theory provides a powerful framework for addressing this challenge \cite{Howard2017,Chitambar2018,albarelli2018resource,Takagi2019,Amaral2019,thomas2024}, revealing in the context of quantum computing how the absence or limited presence of specific quantum features, known as quantum resources, enables efficient classical simulation, and thus shedding light on the requirements for quantum computational advantages.

\begin{figure}[ht!]
    \centering
    \includegraphics[scale = 0.42]{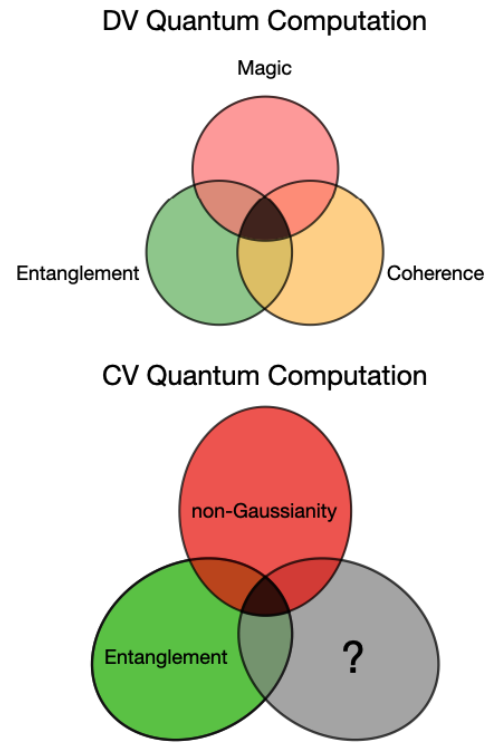}
    \caption{Interplay of resources that determine the complexity of classical simulation of quantum computations, with the black Reuleaux triangle in the middle indicating the regime of universal quantum computation. While universality in DV quantum computations can be seen as resulting from the interplay of magic, entanglement and coherence, whether such an interplay of resources can be seen in a universal CV quantum computation has remained an open question.}
    \label{fig:intro}
\end{figure}

For discrete-variable (DV) quantum computations, such as those based on qubits, quantum computational advantage can be understood as relying on the interplay of three resources: magic, entanglement, and coherence \cite{Horodecki2009,Baumgratz2014,Howard2017,Streltsov2017,Leone2022}. The absence of any one of these resources in a quantum computation enables efficient classical simulation, e.g., using methods based on the Gottesman--Knill theorem \cite{gottesman1998,Aaronson2004} or matrix product states \cite{Vidal2003}.

Continuous-variable (CV) quantum computations \cite{Braunstein2005, Ferraro2005, Weedbrook2012} operate instead on infinite-dimensional Hilbert spaces, where operators such as position and momentum have a continuous spectrum. These computations are implemented in photonic and superconducting systems, platforms which offer enhanced noise robustness \cite{Furusawa1998} and enable unprecedented levels of entanglement \cite{Yokoyama2013}. Notably, these platforms have achieved the experimental milestone of quantum error correction beyond the break-even point using Gottesman–Kitaev–Preskill (GKP) encoding \cite{Gottesman2001,sivak2023}.

\begin{figure*}
    \centering
    \includegraphics[width=0.9\linewidth]{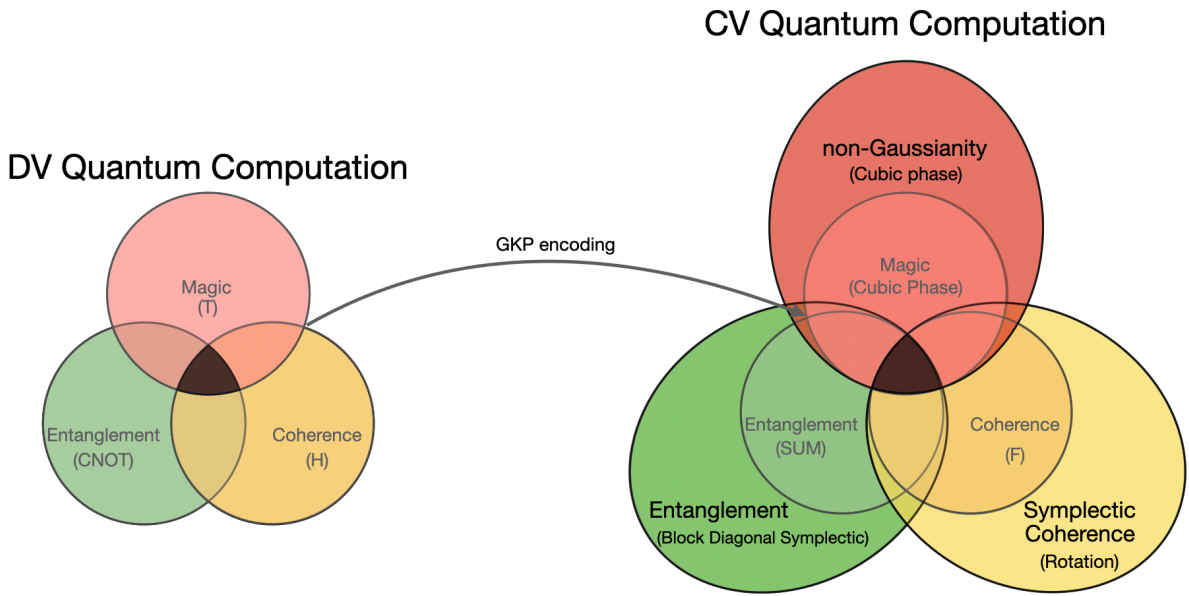}
    \caption{The interplay of CV non-Gaussianity, entanglement, and symplectic coherence generalizes the resource interplay of DV magic, entanglement, and coherence in GKP-encoded quantum computations. Between parentheses are quantum gates inducing each computational resources. The circles on the right hand side denote the resource interplay in the GKP-encoded space within general CV quantum computations. While \textsf{T}, \textsf{CNOT}, \textsf{H}, cubic phase, \textsf{SUM}, and \textsf{F} gates are standard quantum gates in the literature, ``block-diagonal symplectic'' refers to Gaussian unitary gates whose symplectic matrices have a block-diagonal form, i.e., they map position quadratures to position quadratures and momentum quadratures to momentum quadratures, and ``rotation'' refers to gates that may mix position and momentum quadratures within a single mode.}
    \label{fig:conclusion}
\end{figure*}

The GKP encoding maps DV quantum states and operations from a finite-dimensional Hilbert space to a (closure of a) subspace of the infinite-dimensional Hilbert space through the use of logical CV quantum states and operations. As such, it bridges the two computational paradigms and enables established DV simulation techniques to be applied to GKP-encoded CV quantum computations \cite{Alvarez2020, Calcuth2023}.

For general CV quantum computations beyond GKP-encoded ones, the subclass of Gaussian computations is known to be efficiently simulatable on classical computers \cite{Bartlett2002}, underpinning non-Gaussianity as a key quantum resource alongside entanglement \cite{adesso2007entanglement}. Various measures have been developed to quantify non-Gaussianity \cite{Bartlett2002,kenfack2004negativity,Gehrke2012,sperling2015convex,albarelli2018resource,chabaud2020stellar,Marshall2023}, and its role in the complexity of classical simulation has been extensively studied \cite{Mari2012,pashayan2015estimating,chabaud2021simulation,Bourassa2021,chabaud2023resources,Marshall2023,frigerio2024}. However, how the interplay of multiple quantum resources prevents efficient classical simulation and enables computational universality in the CV setting remains underexplored (see Figure \ref{fig:intro}).

In this work, we bridge this gap by identifying three physical resources whose interplay determines the complexity of CV quantum computations: non-Gaussianity, entanglement, and a new resource which we call \textit{symplectic coherence}. Informally, this latter resource is responsible for mixing position and momentum quadratures in CV quantum computations.

We show that symplectic coherence is indeed a quantum computational resource, by introducing a new classical algorithm for simulating CV quantum computations which is efficient for circuits with low symplectic coherence. This algorithm is based on evolving CV quantum operators over simultaneous computational paths (Figure \ref{fig:t-1_orthogonal_reduction}), analogous to recent simulation techniques based on evolution over Pauli paths introduced in the DV setting  \cite{Aharonov2023,fontana2023,angrisani2024,schuster2024,angrisani2025simulating,martinez2025efficient}, but largely unexplored in the CV regime. As a consequence, the absence of any one of the three resources---non-Gaussianity, entanglement, or symplectic coherence---leads to efficient classical simulation, while their combined presence enables universal quantum computation.

Through the GKP encoding, we show that this CV resource interplay can be interpreted as a generalization of the interplay of quantum computational resources observed in DV systems: CV entanglement naturally generalizes DV entanglement, CV non-Gaussianity plays the role of DV magic, as recent literature suggests \cite{Yamasaki2020, Hahn2022, hahn2024}, and CV symplectic coherence corresponds to DV coherence. This result relies on identifying the resources brought by each quantum gate in universal gate sets in DV computations, GKP-encoded CV computations, and general CV computations (see Figure \ref{fig:conclusion}).



The rest of paper is structured as follows: Section \ref{sec:preliminaries} introduces preliminary definitions and notations. Section \ref{sec:symp_coherence} introduces the concept of symplectic coherence and explores its role in enabling universal continuous-variable quantum computation (CVQC). Section \ref{sec:symp_computational} formally establishes symplectic coherence as a computational resource through a new algorithm for efficient classical simulation of CV computations with low symplectic coherence. Section \ref{sec:GKP_interplay} connects resource interplay in DV and CV quantum computations through the GKP encoding, and establishes symplectic coherence as a CV analogue to coherence in the context of quantum computations. Finally, Section \ref{sec:conclusion} highlights the significance of our results and the introduced formalism, and future research directions.


\section{Preliminaries}\label{sec:preliminaries}

Hereafter, the sets $\N, \R$ and $\C$ are the set of natural, real, and complex numbers respectively, with a * exponent when $0$ is removed from the set.
We refer the reader to \cite{NielsenChuang} for quantum information theory material, and to \cite{Braunstein2005,Ferraro2005,Weedbrook2012} for background on CV quantum information, which is the branch of quantum information dealing with observables taking a continuous spectrum of values. A mode refers to the degree of freedom associated with a specific quantum field in a CV system, such as single spatial or frequency mode of light, and is the equivalent of qubit in the CV regime. In this paper, $m \in \N^*$ denotes the number of modes in the system, and $\ket0$ is the vacuum state. $\hat{a}$ and $\hat{a}^\dagger$ refer to the single-mode annihilation and creation operator, respectively, satisfying $[\hat{a},\hat{a}^\dagger] = \I$. $\hat{a}$ and $\hat{a}^\dagger$ are related to the position and momentum quadrature operators as
\begin{equation}
    \hat{q} = \hat{a} + \hat{a}^\dagger, \hspace{5mm}\hat{p} = -i(\hat{a} - \hat{a}^\dagger),
\end{equation} 
with the convention $\hbar = 2$. Furthermore, $\hat{q}$ and $\hat{p}$ satisfy the commutation relation $[\hat{q},\hat{p}] = 2i\I$. The particle number operator is given as $\hat{n} = \hat{a}^\dagger \hat{a}$.

Product of unitary operations generated by Hamiltonians that are quadratic in the quadrature operators of the modes are called Gaussian unitary operations, and states produced by applying a Gaussian unitary operation to the vacuum state are Gaussian states. The action of an $m$-mode Gaussian unitary operation $\hat{G}$ on the vector of quadratures $\boldsymbol{\Gamma}= [\hat{q}_1,\dots,\hat{q}_m,\hat{p}_1,\dots,\hat{p}_m]$ is given by
\begin{equation}\label{eqn:symplectic_transform_quadrature}
    \hat{G}^\dagger \boldsymbol{\Gamma} \hat{G} = S \boldsymbol{\Gamma} + \boldsymbol{d},
\end{equation}
where $S$ is a $2m \times 2m$ symplectic matrix transforming quadratures under $\hat{G}$, and $\boldsymbol{d} \in \R^m$ is a displacement vector. Passive linear operations, induced by beam-splitters and phase shifters in quantum optics, are Gaussian transformations that preserve the total particle number. Among those, orthogonal gates $\hat{O}$ are passive linear unitary operations whose associated symplectic matrix written in the quadrature basis $\boldsymbol{\Gamma} = [\hat{q}_1,\dots,\hat{q}_m,\hat{p}_1,\dots,\hat{p}_m]$ can be put in a block-diagonal form with two  $m\times m$ orthogonal matrices. In particular, they map position quadratures to position quadratures and momentum quadratures to momentum quadratures. On the other hand, rotation gates $\hat{R}$ are passive linear unitary operators that can mix position and momentum quadratures within a mode. Single-mode displacement operators are $\hat{D}(\alpha) = e^{\alpha \hat{a}^\dagger - \alpha^* \hat{a}}$, with $\alpha \in \C$, whereas the shearing gate is given by $e^{is\hat{q}^2}$, with $s\in \R$. 
The Bloch--Messiah decomposition \cite{Ferraro2005} states that any multimode Gaussian unitary gate can be decomposed into a product of orthogonal gates, rotation (or phase-shift) gates, displacement gates and shearing gates.

Non-Gaussian gates are important for enabling quantum advantage since Gaussian gates acting on Gaussian states can be efficiently simulated classically \cite{Bartlett2002}. One prominent example of a non-Gaussian gate is the cubic phase gate $e^{i\gamma\hat{q}^3}$; its action on the quadratures is given by \cite{Budinger2024}
\begin{eqnarray}\label{eq:evo_cubic}
    e^{-i\gamma\hat{q}^3} \hat{q} e^{i\gamma\hat{q}^3}   &=& \hat{q}, \hspace{5mm} e^{-i\gamma\hat{q}^3}  \hat{p}  e^{i\gamma\hat{q}^3}    = \hat{p} + 3\gamma \hat{q}^2.
\end{eqnarray}

A universal gate set comprises gates capable of performing any computation within a given quantum computational model. In the seminal model of continuous-variable quantum computation (CVQC) introduced by Lloyd and Braunstein \cite{Lloyd1999}, the universal gate set consists of Gaussian unitary gates and a single non-Gaussian gate acting on an input vacuum state. Within this framework, universality is defined as the ability to implement any unitary that can be generated by the specific family of polynomial Hamiltonians, i.e., polynomial functions of the position and momentum quadrature operators. 
While this definition may seem restrictive at first sight, recent work demonstrates that any physical unitary operation can be approximated to arbitrary precision by a unitary evolution generated by such a polynomial Hamiltonian \cite{arzani2025}. Since a combination of Gaussian unitary gates and a cubic phase gate can approximate any unitary characterized by a polynomial Hamiltonian with arbitrary precision \cite{Lloyd1999,Sefi2011}, they indeed constitute a universal gate set.
By the Bloch--Messiah decomposition of Gaussian unitary gates, a universal gate set for CVQC can therefore be given as \cite{Menicucci2006}
\begin{equation}
    \{e^{i\gamma\hat{q}_1^3},\hat{O},\hat{R},\hat{D},e^{is\hat{q}_1^2}\},
\end{equation}
where $\hat{O}$, $\hat{R}$, and $\hat{D}$ refer to an orthogonal gate, a rotation gate and a displacement gate, respectively. Note that cubic phase gates and shearing gates are applied on the first mode without loss of generality, up to an orthogonal (SWAP) gate. Further, following \cite{Sefi2011}, each of the shearing gates can be decomposed as displacement gates and cubic phase gates. This reduces the universal gate set for CVQC to
\begin{equation}\label{eqn:ugs}
    \{e^{i\gamma\hat{q}
    _1^3},\hat{O},\hat{R},\hat{D}\}.
\end{equation} 

Note that GKP-encoded computations naturally fit within this CV quantum computing model, as the gate set for this encoding is generated by polynomial Hamiltonians, and the required input GKP states can be approximated up to arbitrary precision using unitary gates generated by polynomial Hamiltonians \cite{arzani2025}.
There is much to be said about the GKP encoding, but in this work we mainly use it as a tool to bridge the DV and CV computational paradigms. The interested reader can refer to \cite{Gottesman2001,Albert2018} for more details on the GKP encoding. In particular, under this encoding, the \textsf{T}, \textsf{CNOT}, and \textsf{H} gate are mapped to the cubic phase gate ($e^{i\gamma\hat{q}^3}$), the $\textsf{SUM}$ gate ($e^{i \frac{\hat{q}_1\hat{p}_2}{2}}$) \cite{Menicucci2006}, and the Fourier transform $\textsf{F}$ gate $(e^{i\frac{\pi}{2}\hat{n}})$, respectively. 

With the necessary material in place, we now define the central concept of this work:\textit{ symplectic coherence}.


\section{Symplectic coherence}\label{sec:symp_coherence}

The evolution of the quantum state of a physical system is governed by a unitary operator which characterizes the system's dynamics. One way to describe this unitary evolution is by specifying its action on the position and momentum quadratures. Gaussian unitary operators transform each quadrature as a symplectic and affine combination of the others, redistributing the information encoded in the state. In contrast, non-Gaussian unitary operators, such as the cubic phase gate in Eq.~(\ref{eq:evo_cubic}), alter the nature of the information by introducing non-linearities.

Focusing on Gaussian unitary operators, recall from Eq.~(\ref{eqn:symplectic_transform_quadrature}) that their action on the quadrature operators is described by a symplectic matrix $S$ and a displacement vector $\bm d$. Hereafter, a Gaussian unitary  operator is said to exhibit \textit{symplectic coherence} if its associated symplectic matrix, expressed in the basis $\boldsymbol{\Gamma} = [\hat{q}_1, \dots, \hat{q}_m, \hat{p}_1, \dots, \hat{p}_m] $, is not in a block-diagonal form consisting of two $m \times m$ matrices, i.e., if it mixes the position and momentum quadratures. The term ``symplectic coherence'' draws an analogy to quantum coherence, which requires the presence of off-diagonal terms in the density matrix of quantum states (in a specific basis). So while a beam-splitter, with efficiency $\eta$ and symplectic matrix
\begin{equation}
    \begin{bmatrix}
        \sqrt{\eta} && \sqrt{1 - \eta} && 0 && 0 \\
        -\sqrt{1 - \eta} && \sqrt{\eta} && 0 && 0 \\
        0 && 0 && \sqrt{\eta} && \sqrt{1 - \eta} \\
        0 && 0 && -\sqrt{1 - \eta} && \sqrt{\eta}
    \end{bmatrix},
\end{equation}
with no non-zero off-block-diagonal terms, lacks symplectic coherence, a Fourier gate with symplectic matrix
\begin{equation} 
\begin{bmatrix}
    0 && 1 \\
    -1 && 0 \\
\end{bmatrix},
\end{equation}
with non-zero off-block-diagonal terms, exhibits it. We define Gaussian unitaries with no non-zero off-block-diagonal terms in their associated symplectic matrix as ``block-diagonal symplectic'' Gaussian unitaries. This terminology aligns with the block-diagonal Gaussian states introduced in \cite{adesso2006generic,Serafini_2007}, which are Gaussian states without direct correlations between position and momentum quadratures. The block-diagonal Gaussian states in \cite{adesso2006generic,Serafini_2007} motivate the exploration of symplectic coherence as a resource in quantum states, defined as the presence of off-diagonal blocks in their covariance matrix. Hereafter, however, we focus on symplectic coherence induced by Gaussian unitary gates.

To understand how this property arises in the context of universal CVQC, we recall the universal gate set for CVQC considered in this work:
\begin{equation}
    \{e^{i\gamma\hat{q}
    _1^3},\hat{O},\hat{R},\hat{D}\}.
\end{equation}
In this gate set, the only Gaussian gates which may introduce symplectic coherence are rotation gates. This raises the question of how the presence of rotation gates inducing symplectic coherence may affect our ability to classically simulate CVQC. The next section addresses this question and demonstrates that the absence of symplectic coherence in a CVQC enables its efficient classical simulation, establishing symplectic coherence as a new computational resource.


\section{Symplectic coherence as a quantum computational resource}\label{sec:symp_computational}

In this section, we consider the problem of computing the expectation value of a Hamiltonian that is polynomial in position and momentum quadratures (a prominent example being the total particle number operator), at the output of a universal CV quantum circuit: given $m,t,d\ge0$ and the descriptions of an $m$-mode input state $\rho$ evolving over an $m$-mode unitary circuit
\begin{figure}
    \centering
    \includegraphics[width=\linewidth]{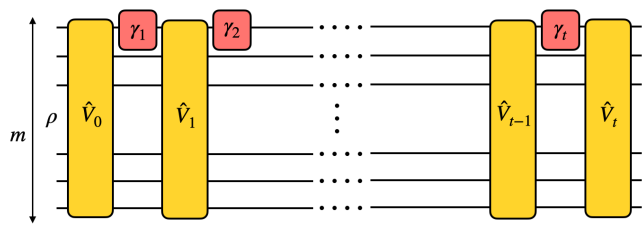}
    \caption{Universal circuit for CV quantum computations. $\gamma_1,\dots,\gamma_t$ are cubic phase gates, whereas $\hat{V}_0,\dots,\hat{V}_t$ are $m$-mode passive linear unitary gates, up to a displacement.}
    \label{fig:og}
\end{figure}
\begin{equation}
    \hat{O}_{\bm{\gamma}} = \hat{V}_t e^{i\gamma_t \hat{q}_1^3} \hat{V}_{t-1} \cdots \hat{V}_1 e^{i\gamma_1 \hat{q}_1^3} \hat{V}_{0},
\end{equation}
where $\bm\gamma=(\gamma_1,\dots,\gamma_t)$, and where $\hat{V}_0,\dots,\hat{V}_t$ are passive linear unitary gates up to a multimode displacement (see Figure \ref{fig:og}), what is the cost of classically computing the expectation value of an operator $H(\hat{\bm{q}},\hat{\bm{p}}) := H(\hat{q}_1,\dots,\hat{q}_m,\hat{p}_1,\dots,\hat{p}_m)$, polynomial of degree $d$ in position and momentum quadratures for the output quantum state $\hat{O}_{\bm{\gamma}}\rho \hat{O}_{\bm{\gamma}}^\dag$ of the circuit? Note that $\hat{O}_{\bm{\gamma}}$ describes a universal circuit for CVQC since displaced passive linear unitary gates (involving orthogonal, rotation and displacement gates) together with cubic phase gates form a universal gate set for CVQC according to Eq.~(\ref{eqn:ugs}).

This problem is generally hard to solve classically: for instance, evolving the position quadrature $\hat{q}_1$ through $t$ cubic phase gates with cubicity $\gamma$ interleaved with Fourier gates produces terms with coefficients of the form $\hat{q}_1^{2^t}$, where repeated squaring after each layer leads to doubly exponential growth \cite{chabaud2024complexity}. Similarly, quadrature evolution in circuits with passive linear unitary gates and cubic phase gates typically results in doubly exponentially large numbers, making classical simulation inefficient. 
However, without symplectic coherence, i.e., if $\hat{V}_0,\dots,\hat{V}_t$ are displaced orthogonal gates denoted by $\hat{O}_0,\dots,\hat{O}_t$, we show that the degree and number of terms remains tractable, and that the circuit can be classically simulated efficiently. This is summarized by the following result:
\begin{theo}\label{theo:classical_simulation_expectation_values_orthogonal}
    Given the descriptions of an $m$-mode input state $\rho$ and an $m$-mode unitary $\hat{O}_{\bm{\gamma}}$ given by
    \begin{equation}\label{eq:full_circ}
    \hat{O}_{\bm{\gamma}} = \hat{O}_t e^{i\gamma_t\hat{q}_1^3} \hat{O}_{t-1} \cdots \hat{O}_1 e^{i\gamma_1\hat{q}_1^3}\hat{O}_0,
    \end{equation}
    where $\hat{O}_0,\dots,\hat{O}_{t}$ are displaced orthogonal gates, the expectation value of an operator $H(\bm{\hat{q}},\bm{\hat{p}})$, polynomial of degree $d$ in the quadrature operators, given by 
\begin{equation}
    \Tr[\hat{O}_{\bm{\gamma}}\rho\hat{O}_{\bm{\gamma}}^\dagger H(\bm{\hat{q}},\bm{\hat{p}}) ] = \Tr[\rho\hat{O}_{\bm{\gamma}}^\dagger H(\bm{\hat{q}},\bm{\hat{p}}) \hat{O}_{\bm{\gamma}}],
\end{equation}
can be computed in time $\mathcal{O}( m^{3d} + t^2 m^6)$.
\end{theo}
\noindent The proof of Theorem \ref{theo:classical_simulation_expectation_values_orthogonal} is given in Appendix \ref{appendixsec:proof_classical_simulation_orthogonal}. It relies on back-propagating single quadrature operators through the circuit $\hat{O}_{\bm{\gamma}}$, which are then combined to get an efficient description of $\hat{O}_{\bm{\gamma}}^\dagger H(\bm{\hat{q}},\bm{\hat{p}}) \hat{O}_{\bm{\gamma}}$, i.e., the back-propagated version of $H({\bm{\hat q}},{\bm{\hat p}})$. This back-propagation of single quadratures is made efficient based on a new classical simulation technique, where instead of evolving a single quadrature over the full unitary circuit in Eq.~(\ref{eq:full_circ}), we simultaneously evolve it over $t$ simpler circuits (computational paths), each involving the displaced orthogonal unitary gates $\hat{O}_0,\dots,\hat{O}_t$, together with a single cubic phase gate, and we sum the results. In this sum over paths, the cubicity of the single cubic phase gate is the same across all paths, and the contribution of each path reflects the fraction of the cubicity brought by the original cubic phase gates in the circuit (see Figure \ref{fig:t-1_orthogonal_reduction}). 
Our quadrature back-propagation technique is formalized in Theorem \ref{theo:t-1_orthogonal_reduction} of the appendix. It is reminiscent of a classical simulation technique for DV quantum computations based on Pauli back-propagation \cite{Aharonov2023,fontana2023,angrisani2024,schuster2024,angrisani2025simulating,martinez2025efficient}, which to the best of our knowledge has not been applied in the CV setting until our work.

Theorem \ref{theo:classical_simulation_expectation_values_orthogonal} establishes symplectic coherence as a computational resource, as its absence enables efficient classical simulation of quantum expectation values of constant-degree polynomial Hamiltonians.
Physically, this result can be understood as follows: a cubic phase gate introduces non-Gaussianity by adding a quadratic term to the momentum quadrature ($\hat{p} \rightarrow \hat{p} + 3\gamma\hat{q}^2$). If the displaced passive linear unitary gates do not mix $\hat{q}$ and $\hat{p}$, the nonlinearity remains confined to a quadratic term in the position quadrature, remaining unaffected by subsequent cubic phase gates. In other words, the displaced passive linear unitary gates do not amplify the non-Gaussianity. However, when passive linear unitary gates do mix $\hat{q}$ and $\hat{p}$, inducing symplectic coherence, the nonlinearity grows with successive layers, making the evolution intractable. 


\begin{figure}
    \centering
    \includegraphics[scale = 0.35]{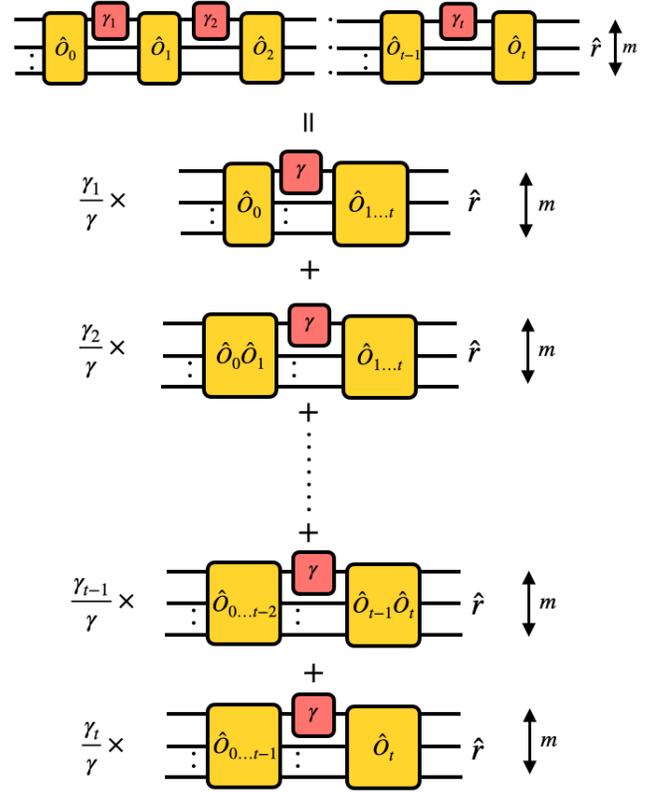}
    \caption{Quadrature back-propagation in the absence of symplectic coherence,  where $\hat{r} \in \{\hat{q}_1,\dots,\hat{q}_m,\hat{p}_1,\dots,\hat{p}_m\}$. The evolution (through back-propagation) of a quadrature over a circuit with a displaced orthogonal unitary between cubic phase gates can be seen as simultaneous evolution over $t$ circuits with a single cubic phase gate and displaced passive linear unitary. Here $\gamma = \sum_{i=1}^t \gamma_i$ and $\hat{O}_{i\dots j} = \hat{O}_i\dots \hat{O}_j$, for $j > i$ and $j,i \in \{0,\dots,t \}$.}
    \label{fig:t-1_orthogonal_reduction}
\end{figure}

We have shown how CV quantum computations without symplectic coherence may be simulated classically. We now further explore whether circuits with low but non-zero symplectic coherence can also be efficiently simulated, similar to CV circuits with low non-Gaussianity \cite{Mari2012,Pashayan2020,chabaud2021simulation,chabaud2023resources} or low entanglement \cite{Nagai2021}. We address this in the following theorem:
\begin{theo}\label{theo:classical_simulation_expectation_values_small_symp}
    Given the description of an $m$-mode input state $\rho$ and an $m$-mode unitary $\hat{U}_{t}$ given by
    \begin{equation}\label{eq:circuit_gen}
    \hat{U}_t = \hat{O}_{\bm{\gamma}_c} \hat{R}_1(\theta_c) \hat{O}_{\bm{\gamma}_{c-1}}\cdots \hat{R}_1(\theta_1) \hat{O}_{\bm{\gamma}_0},
    \end{equation}
    evolving the input state, where $\hat{O}_{\bm{\gamma}_i}$ describes $t + 1$ displaced orthogonal gates interleaved with $t$ cubic phase gates as in Theorem~\ref{theo:classical_simulation_expectation_values_orthogonal}, and $\hat{R}_1(\theta_i)$ are single-mode rotation gates in the first mode $\forall i \in \{0,1,\dots,c \}$, the expectation value of an operator $H(\bm{\hat{q}},\bm{\hat{p}})$, polynomial of degree $d$ in the quadrature operators, given by 
\begin{equation}
    \Tr[\hat{U}_t\rho\hat{U}_t^\dagger H(\bm{\hat{q}},\bm{\hat{p}}) ] = \Tr[\rho\hat{U}_t^\dagger H(\bm{\hat{q}},\bm{\hat{p}}) \hat{U}_t],
\end{equation}
can be computed in time $\mathcal{O}(m^{d2^{c+1}}+(c+1)t^2 m^7)$.
\end{theo}
\noindent Note that rotations in Eq.~(\ref{eq:circuit_gen}) act on the first mode without loss of generality, as the rotation gates on other modes can be shifted to the first one by an orthogonal (\textsf{SWAP}) gate, which can then be absorbed into the corresponding block of displaced orthogonal unitary gates.
The proof of Theorem \ref{theo:classical_simulation_expectation_values_small_symp} is detailed in Appendix \ref{appendixsec:proof_symp}.
As for Theorem \ref{theo:classical_simulation_expectation_values_orthogonal}, it proceeds by back-propagating $2m$ single quadrature operators through the unitary circuits, and multiplying them quadratures to obtain the back-propagated version of $H(\bm{\hat{q}},\bm{\hat{p}})$ under the unitary. However, during the back-propagation, the degree of a single-mode quadrature doubles at every layer due to the rotation gate mixing position and momentum quadratures, leading to a degree of $2^{c+1}$ after $c+1$ layers. At each layer, we sum over all computational paths to express the evolution as a single polynomial, using Theorem \ref{theo:classical_simulation_expectation_values_orthogonal}. This results in a total of $(c+1)\times t$ computational paths, allowing each final back-propagated quadrature operator to be obtained in time $\mathcal{O}(m^{3 \cdot 2^c} + (c+1)t^2 m^7)$. 
Finally, multiplying the back-propagated quadratures introduces an extra time complexity resulting in the one given in the theorem.

In particular, when  $d = \mathcal{O}(1)$, $t = \mathcal{O}(\mathrm{poly}\,m)$ and $c = \mathcal{O}(1)$, the circuit remains efficiently simulable. This demonstrates that CV quantum circuits with low symplectic coherence can also be simulated efficiently classically, thereby strengthening the role of symplectic coherence as a computational resource in CVQC.


\section{CV vs DV resource interplay}\label{sec:GKP_interplay}

Based on the gate set $\{e^{i\gamma\hat{q}_1^3},\hat{O},\hat{R},\hat{D}\}$ in Eq.~(\ref{eqn:ugs}), we have found that the computational complexity of universal CV quantum circuits arises from the interplay of three computational resources: non-Gaussianity (introduced by non-Gaussian gates such as the cubic phase gate), entanglement (generated by block-diagonal symplectic gates such as orthogonal gates), and symplectic coherence (induced by rotation gates, such as the Fourier gate). Note that, while we considered orthogonal gates as the one inducing entanglement in the universal gate set, any block-diagonal Gaussian unitary gate can be taken as entanglement-inducing gate without changing the complexity of the simulation algorithms obtained in the previous section. While displacement gates are also necessary for universal CV quantum computations, they are not expected to impact the complexity of simulating CVQC as their action is analogous to applying a ``classical'' force that shifts the position and momentum quadratures of a quantum harmonic oscillator, without introducing fundamentally quantum effects.

It is natural to ask how this interplay of CV resources relates to that observed in DV quantum computations. In the latter case, computational universality stems from the interaction of magic (induced by gates such as \textsf{T}), entanglement (induced by gates such as \textsf{CNOT}), and coherence (induced by gates such as Hadamard) (see left hand side of Figure \ref{fig:conclusion}). The GKP encoding, which maps DV states and operations to a CV subspace, provides a useful way of comparing the DV and CV resource interplays. Within the GKP code space, the emergence of computational complexity can be understood directly using the DV resource-theoretic framework, but extending this interplay beyond GKP-encoded CVQC is less straightforward. Indeed, while the connections between CV non-Gaussianity and DV magic \cite{Yamasaki2020, Hahn2022, hahn2024}, and between mode-entanglement and logical qubit entanglement \cite{Menicucci2006}, are well-understood within the GKP encoding \cite{Gottesman2001}, a CV computational resource analogous to DV coherence is missing. As we shall explain in what follows, symplectic coherence fills this gap.

To better understand the relationship between the DV and CV resource interplays, we consider each DV resource, specify the corresponding DV gate inducing it within the gate set $\{\textsf{T},\textsf{CNOT},\textsf{H}\}$, obtain its GKP-encoded counterpart, and identify the CV family of gates this GKP-encoded gate belongs to. Finally, we give the CV resource tied to that family of gates. The goal is to determine which CV computational resource is uniquely provided by each gate in the universal GKP-encoded gate set. For magic, the associated DV gate is the \textsf{T} gate, which maps to the cubic phase gate in the GKP-encoded model. The cubic phase gate, in turn, introduces non-Gaussianity in CV computations. Next, entanglement in DV systems, introduced by the \textsf{CNOT} gate, corresponds to the two-mode \textsf{SUM} gate in the GKP encoding. The symplectic matrix of the \textsf{SUM} gate in the quadrature basis $[\hat{q}_1,\hat{q}_2,\hat{p}_1,\hat{p}_2]$ is \cite{Gottesman2001} 
\begin{equation}
    \textsf{SUM} = \begin{bmatrix}
        1 & 0 & 0 & 0 \\
        1 & 1 & 0 & 0 \\
        0 & 0 & 1 & -1 \\
        0 & 0 & 0 & 1
    \end{bmatrix}.
\end{equation}
Therefore, the \textsf{SUM} gate belongs to the class of block-diagonal symplectic Gaussian unitary gates and induces entanglement in CV computations. Finally, coherence in DV systems, induced by the Hadamard gate, corresponds to the Fourier transform (\textsf{F}) gate in the GKP encoding. The symplectic matrix of $\textsf{F}$ gate in the quadrature basis $[\hat{q}_1,\hat{{p}}_1]$ is given by:
\begin{equation}
    \textsf{F} = \begin{bmatrix}
        0 & 1 \\
        -1 & 0
    \end{bmatrix},
\end{equation}
Therefore, \textsf{F} gate belongs to the class of rotation gates and induces symplectic coherence in CV computations.

In conclusion, the interplay of magic, entanglement, and coherence in GKP-encoded DV quantum computations maps to that of non-Gaussianity, entanglement, and symplectic coherence in CV computations, identifying symplectic coherence as a CV analogue of coherence in the context of universal quantum computations, as summarised in Figure \ref{fig:conclusion}.


\section{Conclusion}\label{sec:conclusion}

We have introduced symplectic coherence, a physical property of a unitary responsible for mixing the position and momentum quadratures of the modes. After identifying the role of rotation gates in inducing it, we have demonstrated its importance as a resource for universal CVQC. Specifically, we have derived an efficient simulation algorithm for computing expectation values at the output of CV circuits of low symplectic coherence, i.e. with a constant number of rotation gates. Our proof leverages a novel CV path simulation technique based on quadrature back-propagation. Finally, by connecting CVQC and DVQC using the GKP encoding, we have explained how symplectic coherence completes the framework for understanding the interplay of computational resources enabling universal CVQC, summarised in Figure \ref{fig:conclusion}. Through the GKP encoding, we showed that the CV resource interplay generalizes the DV resource interplay, with CV entanglement, non-Gaussianity, and symplectic coherence corresponding to DV entanglement, magic, and coherence, respectively.
 
Our work opens up several research directions. Future research could investigate whether the absence of symplectic coherence enables efficient classical sampling and strong simulation, on top of expectation value simulation. Moreover, in the DV regime, circuits with at most a logarithmic number of coherence-inducing gates, such as Hadamard gates (or Fourier gates in GKP-encoded computations), are known to be efficiently simulatable \cite{Braun2006,Hillery2016,Shi2017,thomas2024}; we have proven that the same holds for CV quantum circuits with at most a constant number of gates inducing symplectic coherence. Whether a similar result can be shown to hold for a logarithmic number of such gates remains open. We expect this to be the case with an improved simulation algorithm, but it is beyond this paper's scope. Moreover, the path simulation formalism introduced here has provided new insights into quantum complexity in the noiseless case. Given its success in analyzing noisy DV circuits \cite{Aharonov2023,fontana2023,schuster2024,angrisani2025simulating,martinez2025efficient}, extending it to noisy CVQC is a natural next step. Finally, another direction is to formalize symplectic coherence for quantum states and operations within a resource-theoretic framework, and study its operational role beyond quantum computing. We leave these questions for future exploration.


\section{Acknowledgments}
We acknowledge inspiring discussions with  F.\ A.\ Mele, S.\ Oliviero, A.\ Angrisani, and Z.\ Holmes, and funding from the European Union’s Horizon Europe Framework Programme (EIC Pathfinder Challenge project Veriqub) under Grant Agreement No.\ 101114899.

\bibliography{apssamp}
\newpage
\onecolumngrid
\appendix

\section{Proof of Theorem \ref{theo:classical_simulation_expectation_values_orthogonal}}\label{appendixsec:proof_classical_simulation_orthogonal}
\noindent We restate the Theorem first:
\begin{theo}
    Given the description of an $m$-mode input state $\rho$ and an $m$-mode unitary $\hat{O}_{\bm{\gamma}}$ given by
    \begin{equation}
    \hat{O}_{\bm{\gamma}} = \hat{O}_t e^{i\gamma_t\hat{q}_1^3} \hat{O}_{t-1} \dots \hat{O}_1 e^{i\gamma_1\hat{q}_1^3}\hat{O}_0,
    \end{equation}
    evolving the input state, where $\hat{O}_1,\dots,\hat{O}_{t-1}$ are displaced orthogonal gates, the expectation value of an operator $H(\bm{\hat{q}},\bm{\hat{p}})$, polynomial of degree $d$ in the quadrature operators, given by 
\begin{equation}
    \Tr[\hat{O}_{\bm{\gamma}}\rho\hat{O}_{\bm{\gamma}}^\dagger H(\bm{\hat{q}},\bm{\hat{p}}) ] = \Tr[\rho\hat{O}_{\bm{\gamma}}^\dagger H(\bm{\hat{q}},\bm{\hat{p}}) \hat{O}_{\bm{\gamma}}],
\end{equation}
can be computed in time $\mathcal{O}(m^{3d} + t^2 m^6)$.
\end{theo}
\begin{proof}
    We first give the theorem for the quadrature back-propagation technique described in the main text:
    \begin{theo}\label{theo:t-1_orthogonal_reduction}
    With $\hat{r}\in \{\hat{q}_1,\hat{p}_1,\dots,\hat{q}_m,\hat{p}_m\}$ and $\hat{O}_0,\dots,\hat{O}_t$ being $m$-mode displaced orthogonal unitary gates, the evolution of $\hat{r}$ over $\hat{O}_te^{i\gamma_t \hat{q_1}^3} \hat{O}_{t-1} \dots \hat{O}_1 e^{i \gamma_1 \hat{q}_1^3} \hat{O}_0$ can be seen as the weighted sum of its evolution over $t$ unitary gates with the displaced orthogonal unitary gates and a single (modified) cubic phase gate (Figure \ref{fig:t-1_orthogonal_reduction}). That is,
        \begin{equation}
        \hat{O}_0^\dagger e^{-i \gamma_1 \hat{q}_1^3} \hat{O}_1^\dagger \dots\hat{O}_{t-1}^\dagger e^{-i\gamma_t \hat{q_1}^3} \hat{O}_t^\dagger \hat{r} \hat{O}_t e^{i\gamma_t \hat{q_1}^3} \hat{O}_{t-1} \dots \hat{O}_1 e^{i \gamma_1 \hat{q}_1^3} \hat{O}_0 = \sum_{i=1}^t \frac{\gamma_i}{\gamma} \left( \hat{O}_{i-1\dots0}^\dagger e^{-i \gamma \hat{q}_1^3}\hat{O}_{t \dots i}^\dagger \hat{r} \hat{O}_{t \dots i} e^{i\gamma\hat{q}_1^3} \hat{O}_{i-1\dots0} \right),
    \end{equation}
where we have defined
\begin{equation}\label{eqn:brevity_orthogonal}
    \hat{O}_{j\dots i} := \hat{O}_j \dots \hat{O}_i,
\end{equation}
for $i,j \in \{1,\dots,t \}$ and $j>i$, and $\gamma := \sum_{i=1}^t \gamma_i$.
\end{theo}
\begin{proof}
    We will prove that this relation holds for position and momentum quadratures separately.

For $\hat{q} \in \{ \hat{q}_1,\dots,\hat{q}_m\}$,
\begin{equation}
    (\hat{O}_0^\dagger e^{-i \gamma_1 \hat{q}_1^3} \hat{O}_1^\dagger \dots e^{-i\gamma_{t-1}\hat{q}_1^3} \hat{O}_{t-1}^\dagger e^{-i\gamma_t \hat{q_1}^3} \hat{O}_t^\dagger) \hat{q} (\hat{O}_t e^{i\gamma_t \hat{q_1}^3} \hat{O}_{t-1} e^{i\gamma_{t-1}\hat{q}_1^3}\dots \hat{O}_1 e^{i \gamma_1 \hat{q}_1^3} \hat{O}_0) = \hat{O}_{t \dots 0}^\dagger \hat{q} \hat{O}_{t \dots 0},
\end{equation}
since the position quadratures are unaffected by the cubic phase gates and the orthogonal gates mixes the position quadratures with position quadratures. Writing
\begin{equation}
    \hat{O}_{t \dots 0}^\dagger \hat{q} \hat{O}_{t \dots 0} = \frac1{\gamma}\sum_{i=1}^t \gamma_i \hat{O}_{t \dots 0}^\dagger \hat{q} \hat{O}_{t \dots 0},
\end{equation}
the Theorem is proved for $\hat{q}$. 

Let $S_i$ be the $2m \times 2m$ symplectic matrix that characterizes the transformation of quadratures under $\hat{O}_i$. The entry $(S_i)_{q_i,p_j}$ represents the coefficient of $\hat{p}_j$ in the transformation of $\hat{q}_i$ under $\hat{O}_i$, with similar interpretations for other quadratures. Then, given $\hat{p} \in \{ \hat{p}_1,\dots,\hat{p}_m\}$, its evolution follows as
\begin{eqnarray}
    &&(\hat{O}_0^\dagger e^{-i \gamma_1 \hat{q}_1^3} \hat{O}_1^\dagger \dots e^{-i\gamma_{t-1}\hat{q}_1^3}\hat{O}_{t-1}^\dagger e^{-i\gamma_t \hat{q_1}^3} \hat{O}_t^\dagger) \hat{p} (\hat{O}_t e^{i\gamma_t \hat{q_1}^3} \hat{O}_{t-1} e^{i\gamma_{t-1}\hat{q}_1^3} \dots \hat{O}_1 e^{i \gamma_1 \hat{q}_1^3} \hat{O}_0) \nonumber \\
    &&\hspace{5mm} = (\hat{O}_0^\dagger e^{-i \gamma_1 \hat{q}_1^3} \hat{O}_1^\dagger \dots e^{-i\gamma_{t-1}\hat{q}_1^3}\hat{O}_{t-1}^\dagger)( \hat{O}_t^\dagger \hat{p} \hat{O}_t + 3\gamma_t (S_t)_{p,p_1} \hat{q}_1^2) (\hat{O}_{t-1} e^{i\gamma_{t-1}\hat{q}_1^3} \dots \hat{O}_1 e^{i \gamma_1 \hat{q}_1^3} \hat{O}_0) \nonumber \\
    &&\hspace{5mm} = (\hat{O}_0^\dagger e^{-i \gamma_1 \hat{q}_1^3} \hat{O}_1^\dagger \dots \hat{O}_{t-2}^\dagger) (\hat{O}_{t\dots t-1}^\dagger \hat{p} \hat{O}_{t\dots t-1} + 3\gamma_{t-1} (S_{t-1} S_t)_{p,p_1} \hat{q}_1^2) (\hat{O}_{t-2} \dots \hat{O}_1 e^{i \gamma_1 \hat{q}_1^3} \hat{O}_0) + 3\gamma_t (S_t)_{p,p_1} \hat{O}_{t-1\dots0}^\dagger \hat{q}_1^2 \hat{O}_{t-1\dots0}\nonumber\\
    &&\hspace{50mm} \vdots \nonumber\\
    &&\hspace{5mm}= \hat{O}_{t\dots0}^\dagger \hat{p} \hat{O}_{t\dots0} + \sum_{i=1}^t 3 \gamma_i (S_i\dots S_t)_{p,p_1} \hat{O}_{i-1\dots0}^\dagger \hat{q}_1^2 \hat{O}_{i-1\dots0}.
\end{eqnarray}
Since the cubic phase gate affects only the \(\hat{p}_1\) quadrature, transforming it as \(\hat{p}_1 \to \hat{p}_1 + 3\gamma\hat{q}_1^2\), and \(\hat{q}_1^2\) remains unchanged by subsequent cubic phase gates. Therefore, the resulting evolution remains quadratic in the position quadratures. Writing 
\begin{equation}
    \hat{O}_{t\dots0}^\dagger \hat{p} \hat{O}_{t\dots0} = \frac1{\gamma}{\sum_{i=1}^t \gamma_i \hat{O}_{t\dots0}^\dagger \hat{p} \hat{O}_{t\dots0}},
\end{equation}
and grouping $\frac{\gamma_i}{\gamma} \hat{O}_{t\dots0}^\dagger \hat{p} \hat{O}_{t\dots0}$ with $3 \gamma_i (S_i\dots S_t)_{p,p_1} \hat{O}_{i-1\dots0}^\dagger \hat{q}_1^2 \hat{O}_{i-1\dots0}$ for all $i\in\{1,\dots,t\}$, we get the statement of the Theorem.
\end{proof}
\noindent The quadrature back-propagation technique allows us to evolve the single-degree quadratures in $H(\boldsymbol{\hat{q}},\boldsymbol{\hat{p}})$ as degree 2 polynomials. This is formalized in the following Lemma:
\begin{lem}\label{lem:quad_evo_poly}
    Given an $m$-mode unitary 
    \begin{equation}\label{eqn:U_t}
        \hat{O}_{\bm{\gamma}} = \hat{O}_t e^{i\gamma_t\hat{q}_1^3} \hat{O}_{t-1} \dots \hat{O}_1 e^{i\gamma_1\hat{q}_1^3} \hat{O}_0,
    \end{equation} 
    such that $\hat{O}_0,\dots,\hat{O}_{t}$ are $m$-mode displaced orthogonal unitary gates, the evolution of a quadrature $\hat{r} \in \{\hat{q}_1,\dots,\hat{q}_m,\hat{p}_1,\dots,\hat{p}_m\}$ over $\hat{O}_{\bm{\gamma}}$ is described by a quadratic degree polynomial $f_{\hat{r}}^{(2)}(\bm{\hat{q}},\bm{\hat{p}},\hat{O}_{\bm{\gamma}})$, which can be computed in time $\mathcal{O}(t^2 m^5)$.
\end{lem}
\begin{proof}
    Lemma \ref{lem:quad_evo_poly} follows from the fact that a quadrature evolution according to Theorem \ref{theo:t-1_orthogonal_reduction} simplifies to a quadrature evolution in $t$ circuits with $(t+1)$ orthogonal unitary gates and a single cubic phase gate. The evolution of the quadrature along each of these simplified circuits is described as a degree two polynomial with $\mathcal{O}(m^2)$ terms. Given the description of $\hat{O}_0,\dots,\hat{O}_t$, the description of $\hat{O}_i \times \dots \times \hat{O}_j$ can be calculated in time $\mathcal{O}((j-i)m^3)$. Since each of the $t$ simplified circuits have $t+1$ orthogonal unitary gates, this implies that the degree 2 polynomial along each of the paths can be given in time $\mathcal{O}((t+1)m^3)$. Further, since we have to sum the coefficient of each of the $\mathcal{O}(m^2)$ terms in $t$ such polynomials, this gives the time complexity of computing $f_{\hat{r}}^{(2)}(\bm{\hat{q}},\bm{\hat{p}},\hat{O}_{\bm{\gamma}})$ to be of the order $\mathcal{O}(t(t+1)m^5) \sim \mathcal{O}(t^2 m^5)$ upto lower degree polynomial factors in $t$.
\end{proof}
With Lemma \ref{lem:quad_evo_poly}, Theorem \ref{theo:classical_simulation_expectation_values_orthogonal} can be easily seen by noting that with the quadrature back-propagation picture, the single degree quadratures involved in the description of $H(\bm{\hat{q}},\bm{\hat{p}})$ can be evolved independently as quadratic degree polynomials in time $\mathcal{O}(t^2 m^6)$ which corresponds to the time calculating the evolution of each $\hat{q}_1,\dots,\hat{q}_m,\hat{p}_1,\dots,\hat{p}_m$, as all of them may be involved in the description of $H(\boldsymbol{q},\boldsymbol{p})$. Then we need to multiply $\mathcal{O}(m^{2d})$ terms to get the evolution of terms of degree $d$ in $H(\boldsymbol{q},\boldsymbol{p})$, and similarly for terms of degree $0,1,2,\dots,d-1$. The general form of $H(\bm{\hat{q}},\bm{\hat{p}})$ is
\begin{equation}
    \sum_{\alpha_1 + \alpha_2 + \dots + \alpha_m + \beta_1 + \dots + \beta_m \leq d} \hat{q}_1^{\alpha_1}\dots\hat{q}_m^{\alpha_m} \hat{p_1}^{\beta_1}\dots\hat{p_m}^{\beta_m},
\end{equation}
with $\alpha_1,\dots,\alpha_m,\beta_1,\dots,\beta_m \in \N$. Note that any other ordering of the operators can be put in this form by the application of the appopriate commutation relation. This gives the total number of terms as
\begin{equation}
    \sum_{n=0}^d \binom{2m+n-1}{n} = \binom{2m + d}{d} = \mathcal{O}(m^d),
\end{equation}
where the right hand side follows from the hockey-stick identity. Therefore we have at most $\mathcal{O}(m^d)$ terms and each one of them can be described in time at most $\mathcal{O}(m^{2d})$. This gives the total time to describe $\hat{O}_{\bm{\gamma}}^\dagger H(\bm{\hat{q}},\bm{\hat{p}}) \hat{O}_{\bm{\gamma}}$ to be $\mathcal{O}(m^{3d} + t^2 m^6)$. And since we have the efficient description of the input state $\rho$, this also gives the time complexity of computing its expectation value.
\end{proof}

\section{Proof of Theorem \ref{theo:classical_simulation_expectation_values_small_symp}}\label{appendixsec:proof_symp}
\noindent We restate Theorem \ref{theo:classical_simulation_expectation_values_small_symp}:
\begin{theo}
    Given the description of an $m$-mode input state $\rho$ and an $m$-mode unitary $\hat{U}_{t}$ given by
    \begin{equation}
    \hat{U}_t = \hat{O}_{\bm{\gamma}_c} \hat{R}_1(\theta_c) \hat{O}_{\bm{\gamma}_{c-1}}\dots \hat{R}_1(\theta_1) \hat{O}_{\bm{\gamma}_0},
    \end{equation}
    evolving the input state, where $\hat{O}_{\bm{\gamma}_i}$ describes $t + 1$ displaced orthogonal gates interleaved with $t$ cubic phase gates, and $\hat{R}_1(\theta_i)$ are single-mode rotation gates in the first mode $\forall i \in \{0,1,\dots,c \}$, the expectation value of an operator $H(\bm{\hat{q}},\bm{\hat{p}})$, polynomial of degree $d$ in the quadrature operators, given by 
\begin{equation}
    \Tr[\hat{U}_t\rho\hat{U}_t^\dagger H(\bm{\hat{q}},\bm{\hat{p}}) ] = \Tr[\rho\hat{U}_t^\dagger H(\bm{\hat{q}},\bm{\hat{p}}) \hat{U}_t],
\end{equation}
can be computed in time $\mathcal{O}(m^{d2^{c+1}}+(c+1)t^2 m^7)$.
\end{theo}
\begin{proof}
    The application of Lemma \ref{lem:quad_evo_poly} and Theorem \ref{theo:classical_simulation_expectation_values_orthogonal} for the evolution of a single -degree quadrature over each of $\hat{O}_{\bm{\gamma}_i}$ gives the following Lemma:
    \begin{lem}\label{lem:quad_evo_low_symp}
    Given an $m$-mode unitary 
    \begin{equation}
        \hat{U}_t = \hat{O}_{\boldsymbol{\gamma}_c} \hat{R}_1(\theta_c) \hat{O}_{\boldsymbol{\gamma}_{c-1}}\dots \hat{R}_1(\theta_1) \hat{O}_{\boldsymbol{\gamma}_0},
    \end{equation} 
    such that $\hat{O}_{\boldsymbol{\gamma}_i}$ describes $t + 1$ displaced orthogonal unitary gates interleaved with $t$ cubic phase gates, and $\hat{R}_1(\theta_i)$ are single-mode rotation gates in the first mode $\forall i \in \{0,1,\dots,c \}$, the evolution of a quadrature $\hat{r} \in \{\hat{q}_1,\dots,\hat{q}_m,\hat{p}_1,\dots,\hat{p}_m\}$ over $\hat{U}_t$ is described by a polynomial of degree $2^{c+1}$ and number of terms $\mathcal{O}(m^{2^{c+1}})$, whose description can be given in time $\mathcal{O}(m^{3*2^{c}}+(c+1)t^2 m^6)$
\end{lem}
\begin{proof}This Lemma follows from the iterative structure of the evolution. Starting with quadrature operator $\hat{r} \in \{\hat{q}_1,\dots,\hat{q}_m,\hat{p}_1,\dots,\hat{p}_m\}$, the application of $\hat{O}_{\boldsymbol{\gamma}_c}$ using the quadrature backpropagation technique (Lemma \ref{lem:quad_evo_poly}) transforms it into a polynomial of degree at most two with $\mathcal{O}(m^2)$ terms in the quadratures where the quadratic terms involve only the position quadratures. Using Theorem \ref{theo:classical_simulation_expectation_values_orthogonal}, the description of this polynomial can be computed in time $\mathcal{O}(t^2 m^6 + m^3)$. Applying the rotation $\hat{R}_1(\theta_c)$ to this polynomial produces another polynomial of degree two with $\mathcal{O}((m+1)^2)$ terms. This transformation is computationally efficient since it involves a single rotation gate, and crucially, it introduces a quadratic term of the form $\hat{p}_1^2$ by transforming the quadrature $\hat{q}_1^2$ as $(\cos(\theta_c)\hat{q}_1 + \sin(\theta_c)\hat{p}_1)^2$. 

Therefore, applying $\hat{O}_{\boldsymbol{\gamma}_{c-1}}$ to this polynomial results in an evolution where the previously introduced $\hat{p}_1^2$ transforms as $(\hat{p}_1 + 3 \gamma\hat{q}_1^2)^2$, leading to a polynomial of degree four, with quadratic terms involving only the position quadratures, as the subsequent displaced orthogonal unitaries acting on $\hat{q}_1^4$ do not mix position quadratures with momentum quadratures. Using Theorem \ref{theo:classical_simulation_expectation_values_orthogonal}, the description of this polynomial can be obtained in time $\mathcal{O}(t^2 m^6 + m^6)$. The application of $\hat{R}_1(\theta_{c-1})$ on the quartic term $\hat{q}_1^4$ then introduces a term of the form $\hat{p}_1^4$, and repeating the process with the next layer $\hat{O}_{\boldsymbol{\gamma}_{c-2}}$ results in a polynomial of degree eight. This pattern continues, with each additional layer doubling the polynomial degree. Therefore, by induction, the degree of the evolved single-degree quadrature $\hat{r}$ transforms as
\begin{eqnarray}
    1 \overset{\hat{O}_{\bm{\gamma}_c}}{\longrightarrow} 2 \overset{\hat{R}_1(\theta_c)\hat{O}_{\bm{\gamma}_{c-1}}}{\longrightarrow} 2^2 \overset{\hat{R}_1(\theta_{c-1})\hat{O}_{\bm{\gamma}_{c-2}}}{\longrightarrow} 2^3  \cdots \overset{\hat{R}_1(\theta_{1})\hat{O}_{\bm{\gamma}_{0}}}{\longrightarrow} 2^{c+1},
\end{eqnarray}
whereas the number of terms in the polynomial corresponding to the evolved form of $\hat{r}$ transform as
\begin{eqnarray}
    1 \overset{\hat{O}_{\bm{\gamma}_c}}{\longrightarrow} \mathcal{O}(m^2) \overset{\hat{R}_1(\theta_c)\hat{O}_{\bm{\gamma}_{c-1}}}{\longrightarrow}\mathcal{O}(m^{2^2})\overset{\hat{R}_1(\theta_{c-1})\hat{O}_{\bm{\gamma}_{c-2}}}{\longrightarrow} \mathcal{O}(m^{2^3}) \cdots
    \cdots \overset{\hat{R}_1(\theta_{1})\hat{O}_{\bm{\gamma}_{0}}}{\longrightarrow} \mathcal{O}(m^{2^{c+1}}).
\end{eqnarray}
Therefore, after $c+1$ layers, the single degree quadrature $\hat{r}$ evolves to a polynomial of degree $2^{c+1}$ in the quadratures, with the number of terms of the order $\mathcal{O}(m^{2^{c+1}})$.

Applying Theorem \ref{theo:classical_simulation_expectation_values_orthogonal} at each step, the total time required to describe the evolved quadrature after $(c+1)$ layers is given by $\mathcal{O}((c+1)t^2 m^6 + \sum_{i=0}^{c} m^{3 \cdot 2^i})$, which accounts for the cumulative time needed to describe the evolved polynomial after each of the $c+1$ layers. This expression simplifies to $\mathcal{O}(m^{3 \cdot 2^c} + (c+1)t^2 m^6)$, as the term $m^{3 \cdot 2^c}$ dominates due to the double-exponential growth in $m^{3 \cdot 2^i}$.
\end{proof}
\noindent The complexity of computing the expectation value of the $H(\bm{\hat q},\bm{\hat p})$ is determined by first evolving all $2m$ single degree quadratures into polynomials of degree $2^{c+1}$, each containing $\mathcal{O}(m^{2^{c+1}})$ terms. According to Lemma \ref{lem:quad_evo_low_symp}, this requires $\mathcal{O}(m ( m^{3 \cdot 2^c} + (c+1)t^2 m^6)) = \mathcal{O}(m^{3 \cdot 2^c +1} + (c+1)t^2 m^7)$ time. Next, describing an evolved quadrature of degree $d$ involves multiplying $\mathcal{O}(m^{d 2^{c+1}})$ terms, and $H(\bm{\hat q},\bm{\hat p})$ consists of at most $\mathcal{O}(m^d)$ terms of degree at most $d$. Combining these contributions, given the efficient description of $\rho$, the overall complexity of computing $\Tr[\rho\hat{U}_t^\dagger H(\bm{\hat q},\bm{\hat p}) \hat{U}_t]$ is given by $\mathcal{O}(m^{d + d 2^{c+1}} + m^{3 \cdot 2^c} + (c+1)t^2 m^7)$, which simplifies to $\mathcal{O}(m^{d 2^{c+1}} + (c+1)t^2 m^7)$, up to lower-degree polynomial factors in $m$.

\end{proof}
\end{document}